 \documentclass{lmcs} 
\usepackage{amssymb}

\usepackage{enumitem}
\usepackage{graphics}
\usepackage{subfig}
\usepackage{picins}
\usepackage{amsmath}
\setlist[enumerate]{itemsep=0mm}

\begin{document}
 \title[Logic SPN]{Logic Signed Petri Net}


 \author[Payal]{Payal}
 \address{ Department of Applied Mathematics,
 	Delhi Technological University,
 Bawana Road,Delhi,India- 110042}
 \email{payal.dtu@gmail.com}

 \author[Sangita]{Sangita Kansal}
 \address{Department of Applied Mathematics,
 	Delhi Technological University,
 	Bawana Road,Delhi,India- 110042}
 \email{sangita\_kansal15@rediffmail.com}

\maketitle
\begin{abstract}
	In this paper,the authors show the versatility of the Signed Petri Net (SPN) introduced by them by showing the equivalence between a Logic Signed Petri Net (LSPN) and Logic Petri Net (LPN).
	The capacity of each place in all these nets is at most one, i.e., a place has either zero or one token in it.
\end{abstract}

%


\maketitle

\section{Introduction}
Petri Net is a well-known modeling tool to model concurrent systems.Since its inception, various extensions of PN's have been introduced viz., Continuous PN, Stochastic PN, Timed PN, Object PN, Hybrid PN, Workflow nets, Fuzzy PN, Lending PN, Multidimensional PN \cite{All,Mar,Mur,Wan,Rud,Van,Che,Mass,Van1} etc.The PN and its extensions have been widely used in various fields \cite{HEL,DIF,LIU,CON}.To describe batch processing functions and passing value indeterminacy in cooperative systems, LPNs have been defined as high level PNs in \cite{YDu1}.The authors introduced the concept of Signed Petri net (SPN)\cite{Pay} by utilizing the concept of signed graph and Petri net.SPNs are capable of modeling a large variety of systems and is prefered over a Petri net due to the presence of two types of tokens in it, positive and negative, which are distinguishable.Further, in comparison to a signed graph, SPN is advantageous since a single SPN can be used to represent various signed graphs by simply varying the marking of SPN due to firing of a sequence of transitions.Thus, we need to analyse one SPN in order to infer about all possible signed graph structures that can be formed for a fixed number of vertices.Although an SPN was introduced keeping in view the above mentioned advantages,in this paper, authors show the versatility of the SPN by showing that an equivalent LPN can be constructed using an SPN, which the authors call a Logic Signed Petri Net (LSPN).
\section{Basic Definitions}
\subsection{Petri Net(\cite{Jen})}

A \textit{Petri net (PN)} is a 5-tuple $ N=(P,T,I^{-},I^{+},\mu_0)$,where

\begin{enumerate}
	\item 	P is the finite, non-empty set of places.
	\item 	T is the finite, non-empty set of transitions.
	\item 	$ P\cap T =\emptyset$.
	\item 	$ I^{-},I^{+} :(P\times T)\rightarrow \mathbb{N}$ where $\mathbb{N}$ is the set of non-negative integers, are called negative and positive incidence functions respectively.
	\item 	$\forall p \in P,\exists \ t \in T $ such that $I^{-}(p,t) \neq 0 \ or \ I^{+}(p,t) \neq 0$,and\\
	$\forall t \in T,\exists \ p \in P $ such that $I^{-}(p,t) \neq 0 \ or \ I^{+}(p,t) \neq 0$
	
	\item 	$\mu_0 :P\rightarrow \mathbb{N} $ is the initial marking which gives the initial distribution of tokens in places.
	
	The arc set of the Petri net $N$ is defined as:
	$$E=\{(p,t):I^{-}(p,t)>0\} \cup \{(t,p):I^{+}(p,t)>0\}$$

\end{enumerate}

\subsection{Signed Petri Net(\cite{Pay})}  
\label{S_1}

\begin{defi}{\textbf{Signed Petri Net (SPN)}}
	
	A Signed Petri Net is defined as a 3-tuple $N^{*}=(N',\sigma,\mu_0)$ ,where 
	\begin{enumerate}
		\item $N'=(P,T,I^{-},I^{+})$ is a Petri net structure. 
		\item $\sigma :E \rightarrow \{+,-\}$, where $E$ is the arc set of $N'$.An arc is called a positive or negative arc respectively according to the sign $+$ or  $-$ assigned to it using the function $\sigma$.
		\item  $\mu_0=(\mu^{+}_0,\mu^{-}_0)$ is the initial marking of SPN where
		\begin{enumerate}
			\item $\mu_0^{+}:P\rightarrow \mathbb{N}$ gives the initial distribution of positive tokens in the places,called positive marking of SPN.
			\item $\mu_0^{-}:P\rightarrow \mathbb{N}$  gives the initial distribution of negative tokens in the places,called negative marking of SPN.
		
	\end{enumerate}
	
\end{enumerate}

\end{defi}

Thus,a marking in SPN can be represented as a vector $\mu=(\mu^{+},\mu^{-})$ with $\mu^+,\mu^- \in \mathbb{N}^n,n=|P|$ such that $\mu(p_i)=(\mu^{+}(p_i),\mu^{-}(p_i)) \ \forall \ p_i \in P$.\\
\par Graphically, positive and negative arcs in an SPN are represented by solid and dotted lines respectively.A positive token is represented by a filled circle and a negative token by an open circle.\\
An SPN is said to be \textit{negative} if all of its arcs are negative in sign.

\begin{rem}
	$N''=(N',\sigma)$ is called an SPN structure where $N'$ is a PN structure and $\sigma :E \rightarrow \{+,-\}$.
\end{rem}

\subsubsection{Execution Rules for Signed Petri Net}
\label{prop}
Similar to a Petri net, the execution of an SPN depends on the distribution of tokens in its places.The execution takes place by firing of a transition.A transition may fire if it is enabled.\\
A transition $t$ in an SPN $N^{*}$ is \textit{enabled} at a marking $\mu=(\mu^+,\mu^-)$ if \\
$$
I^{-}(p,t) \leq \mu^{+}(p) \ \forall p \in {}^{\bullet}t \ \textnormal{for which} \ \sigma(p,t) = + $$
$$	 I^{-}(p,t) \leq \mu^{-}(p) \ \forall p \in {}^{\bullet}t \ \textnormal{for which} \ \sigma(p,t) = - 
$$
An enabled transition $t$ may \textit{fire} at $\mu=(\mu^+,\mu^-)$ provided $\exists p_k \in t^{\bullet}$ such that:
$$\sigma(t,p_k)=
\begin{cases}
+ & \text{if}\  \sigma(p,t)=+ \ \forall p \in {}^{\bullet}t \\
- & \text{if}\  \sigma(p,t)=- \ \forall p \in {}^{\bullet}t \\
+ \ \text{or} \ - & \text{if}\  \sigma(p,t)=+ \ \text{for some} \ p \in {}^{\bullet}t \ and \ - \ \text{for some} \ p \in {}^{\bullet}t \\
\end{cases}
$$
After firing, it yields a new marking $\mu_1=(\mu_1^+,\mu_1^-)$ given by the rule:
$$\mu_1^+(p)=\mu^+(p) - I^{-}(p,t) +  I^{+}(p,t) \ \forall p \in P \quad \textnormal{where (p,t) \& (t,p) are positive arcs,}$$ $$ \textnormal{\qquad \qquad \qquad if exist} $$
$$\mu_1^-(p)=\mu^-(p) - I^{-}(p,t) +  I^{+}(p,t) \ \forall p \in P \quad \textnormal{where (p,t) \& (t,p) are negative arcs,}$$ $$ \textnormal{\qquad \qquad \qquad if exist}$$
We say that $\mu_1$ is reachable from $\mu$ and write $\mu \stackrel{t}\to \mu_1$.We restrict the movement of positive(negative) tokens to positive(negative) arcs only.

\par
A marking $\mu$ is reachable from $\mu_0$ if there exists a firing sequence $\eta$ that transforms $\mu_0$ to $\mu$ and is written $\mu_0 \stackrel{\eta}\to \mu$.A \textit{firing or occurence sequence} is a sequence of transitions $\eta=t_1t_2\ldots t_k$ such that 
$$\mu_0 \stackrel{t_1}\to \mu_1\stackrel{t_2}\to \mu_2 \stackrel{t_3} \to \mu_3 \ldots \stackrel{t_k}\to \mu$$
Note that a transition $t_j,1 \leq j \leq k$ can occur more than once in the firing sequence $\eta$.\\
Let us look at the execution of an SPN with the help of an example.\\
In figure \ref{SPN13}$(a)$, \ $t_1$ and $t_2$ both are enabled at $\mu_0$. Firing of $t_1$ yields a new marking $\mu=((0,1,1,0),$$(1,0,1,0))$ and firing of $t_2$ yields a new marking $\mu=((1,0,2,0),(0,1,0,1))$.
In figure \ref{SPN13}$(b)$, \ $t_1$ is enabled,while $t_2$ is not. $t_1 $ can fire to give a new marking
$\mu=((0,0,1,0),(0,0,0,1))$.\\
\begin{figure}[ht]
	\centering
	\subfloat[SPN with $\mu_0=((1,0,1,0),(1,0,0,0))$ ]{{\includegraphics[scale=0.35]{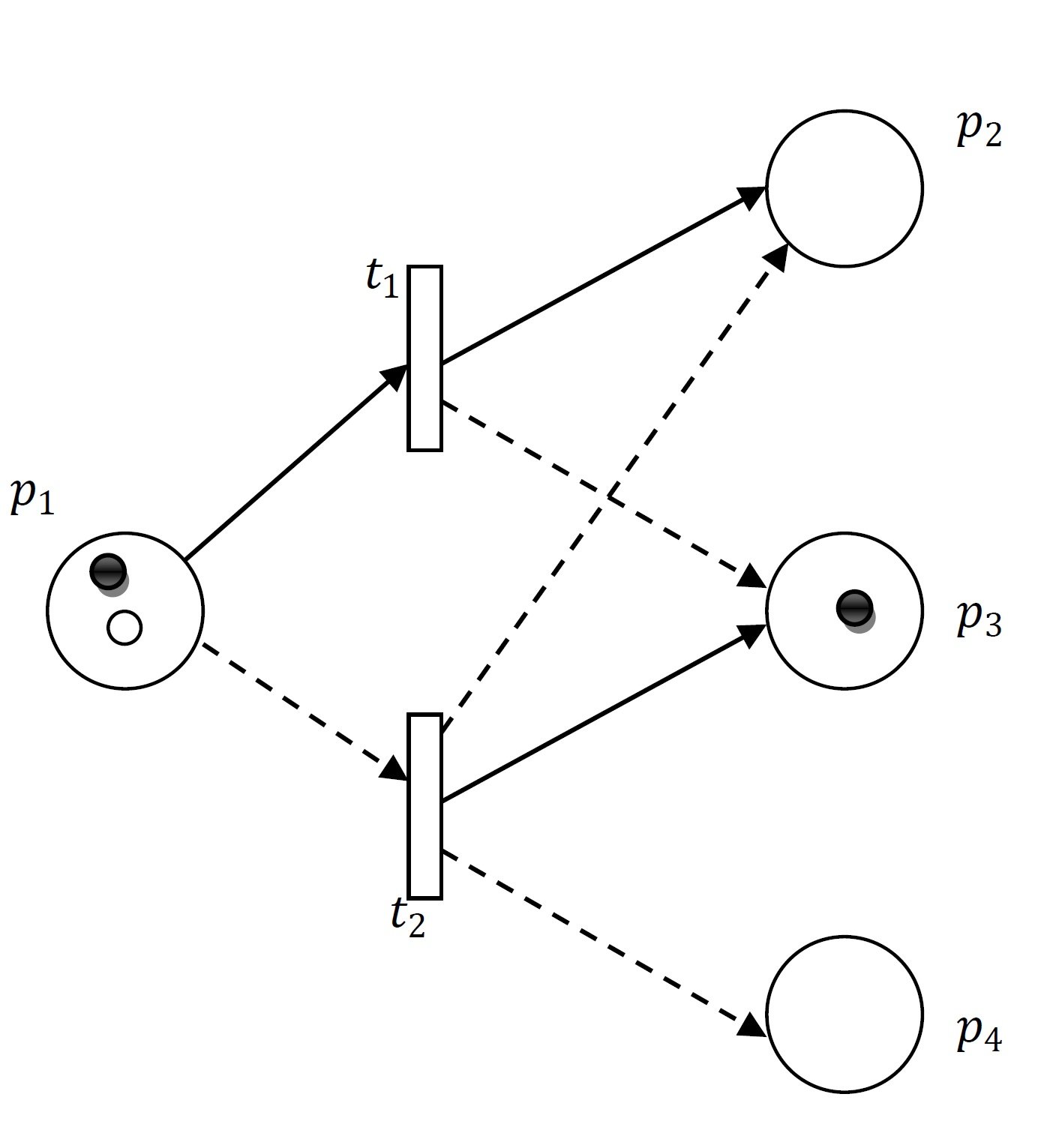} }}
	\qquad \qquad \qquad
	\subfloat[SPN with $\mu_0=((1,0,0,0),(0,0,0,0))$]{{\includegraphics[scale=0.3]{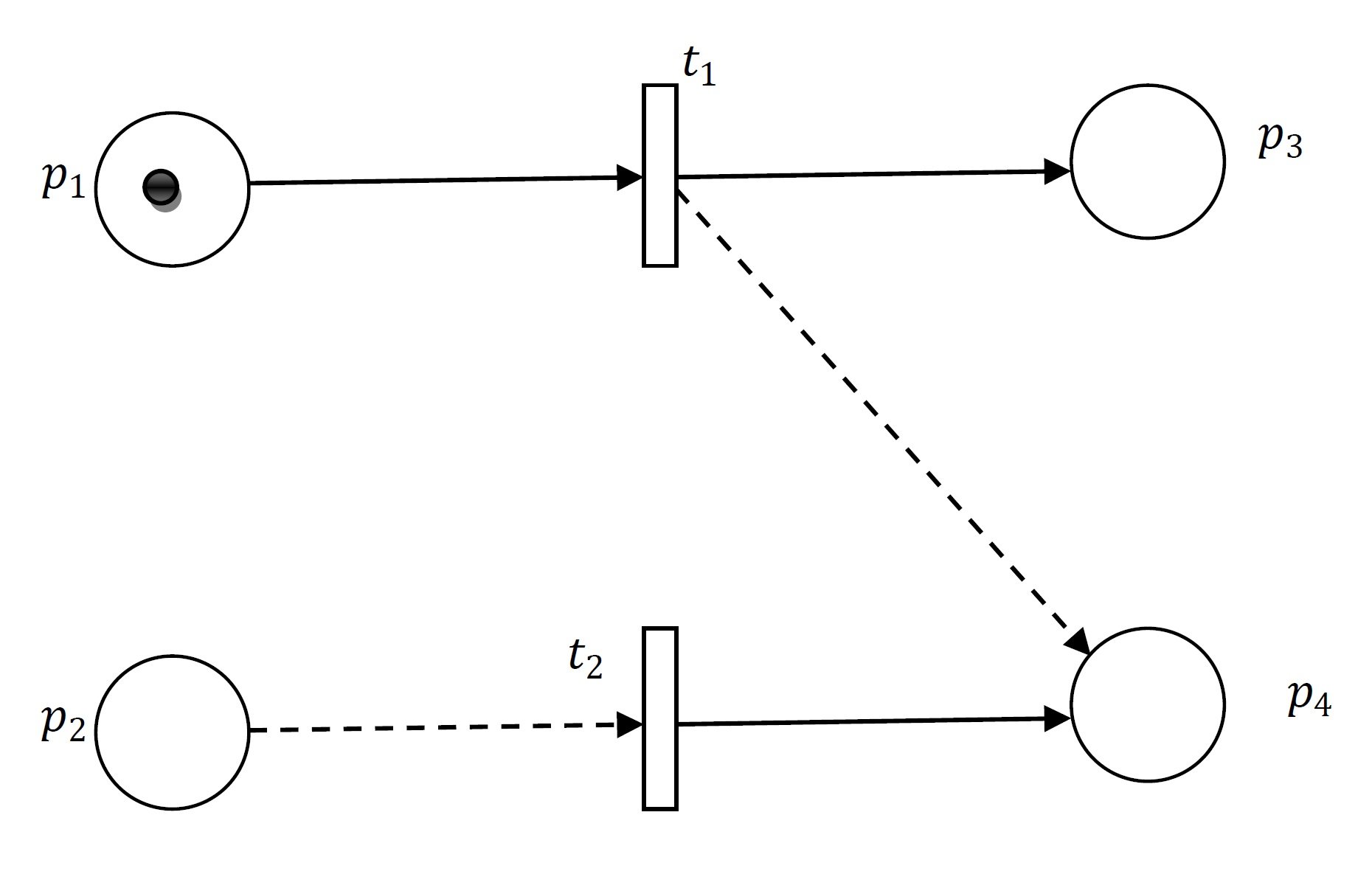} }}
	\caption{ {Execution of an SPN}}
	\label{SPN13}
\end{figure}

\begin{defi}{\textbf{Reachability Set of Signed Petri net}}
	
	The Reachability Set $R(N^{*},\mu_0)$ of an SPN $N^{*}$ is the set of all markings of $N^*$ reachable from the initial marking $\mu_0$.
\end{defi}

\begin{defi}{\textbf{Logic Petri net}}(\cite{YDu1})

	$LPN=(P,T,I^-,I^+,I,O,\mu_0)$ is a logic PN if and only if 
	\begin{enumerate}
		\item $(P,T,I^-,I^+,\mu_0)$ is a PN.
		\item T includes three subsets of transitions,i.e.,$T=T_D\cup T_I \cup T_O , \ \forall \ t \in T_I \cup T_O :{}^\bullet t \cap t^\bullet = \emptyset$ where
		\begin{itemize}
			\item $T_D$ denotes a set of traditional transitions.
			\item $T_I$ denotes a set of logic input transitions.
			\item $T_O$ denotes a set of logic output transitions.
		\end{itemize}
		\item $I$ is a mapping from logic input transitions to a logic input expression,i.e.,$\forall \ t \in T_I,I(t)=f_I(t)$, where $f_I(t)$ is a logic input expression associated with the transition $t$.
		\item $O$ is a mapping from logic output transitions to a logic output expression,i.e.,$\forall \ t \in T_O,O(t)=f_O(t)$, where $f_O(t)$ is a logic output expression associated with the transition $t$.
		\item \textit{Transition Firing Rules}
		\begin{itemize}
			\item $\forall t \in T_D$,the firing rules are same as in PNs.
			\item $\forall t \in T_I$, $t$ is enabled at $\mu$ if the input expression $f_I(t)$ is true at $\mu$.After firing it yields a new marking $\mu'$ given by the following rule:\\
			$\forall p \in {}^\bullet t,$ if $ \mu(p)=1,$ then $\mu'(p)=\mu(p)-1;$\\
			$\forall p \in t^\bullet ,\mu'(p)=\mu(p)+1$\\
			$\forall p \notin {}^\bullet t \cup t^\bullet,\mu'(p)=\mu(p)$
			\item $\forall t \in T_O$,t is enabled if $\forall p \in {}^\bullet t:\mu(p)=1.$ Firing $t$ will generate a new marking $\mu',$ given by: 
			$\forall p \in S: \mu'(p)=\mu(p)+1;$\\
			$\forall p \in {}^\bullet t ,\mu'(p)=\mu(p)-1$\\
			$\forall p \notin {}^\bullet t \cup t^\bullet$ or $ \forall p \in t^\bullet \backslash S: \mu'(p)=\mu(p)$	
			where $ S\subseteq t^\bullet,$ such that output expression $f_O(t)$ is true at $\mu'$.
		\end{itemize}

	\end{enumerate}

\end{defi}

\begin{rem}
	In the definition of Logic Petri Net, a logic input expresssion is attached to a logic input transition.By this we mean a logic expression involving the input places of the given transition.\\
	Similarly, a logic output expresssion is attached to a logic output transition.By this we mean a logic expression involving the output places of the given transition.

\end{rem}

\section{Logic Signed Petri Net(LSPN)}

The Logic Signed Petri Net (LSPN) has been obtained by modification of execution rules of an SPN as defined in \cite{Pay}.
\begin{defi}{\textbf{Logic Signed Petri net}}
	
	A logic SPN is defined as an SPN $N^*=(P,T,A^-,A^+,B^-,B^+,\mu_0)$ where 
	\begin{enumerate}[(i)]
		\item $P$ is a finite,non empty set of places.
		\item $T$ includes three subsets of transitions,i.e.,$T=T_D\cup T_I \cup T_O , \ \forall \ t \in T_I \cup T_O :{}^\bullet t \cap t^\bullet = \emptyset$ where
		\begin{itemize}
			\item $T_D$ denotes a set of traditional transitions.
			\item $T_I$ denotes a set of logic input transitions.
			\item $T_O$ denotes a set of logic output transitions.
		\end{itemize}
		If $t \in T_I$ then $\forall \ p \ \in {}^\bullet t $ either $\sigma(p,t)=+$ or $\sigma(p,t)=+$ as well as $-$ and $\forall \ p \ \in t^\bullet$ $\sigma(t,p)=+$.\\
		Similarly, if $t \in T_O$ then $\forall \ p \ \in t^\bullet  $ either $\sigma(t,p)=+$ or $\sigma(t,p)=+$ as well as $-$ and $\forall \ p \ \in {}^\bullet t $ $\sigma(p,t)=+$.
		\item $A^- =[a^-_{ij}]$ where $a^-_{ij}$ gives the number of positive arcs from $p_j$ to $ t_i$.
		\item $A^+ =[a^+_{ij}]$ where $a^+_{ij}$ gives the number of positive arcs from $ t_i$ to $p_j$.
		\item $B^- =[b^-_{ij}]$ where $b^-_{ij}$ gives the number of negative arcs from $p_j$ to $t_i$.
		\item $B^+ =[b^+_{ij}]$ where $b^+_{ij}$ gives the number of negative arcs from $ t_i$ to $p_j$.
		\item A marking in an LSPN can be represented as a vector $\mu=(\mu^{+},\mu^{-})$ with $\mu^+,\mu^- \in \mathbb{N}^n,n=|P|$ such that $\mu(p_i)=(\mu^{+}(p_i),\mu^{-}(p_i)) \ \forall \ p_i \in P$.\\
		
		\subsection{Execution rules for LSPN}
		\begin{enumerate}[(i)]
			\item If $t_i \in T_D$, then the execution rules are same as for an SPN.
			\item For $t_i \in T_I$ 
			\begin{itemize}
				\item \textit{Enabling Condition}- A transition $t_i\in T_I$ is enabled at $\mu$ provided
				\begin{enumerate}
					\item $\mu^+(p_j)=1 \ \forall \ p_j \ \in \ S_1=\{ p_j  \in {}^\bullet t_i \ | \ a_{ij}^-=1,b_{ij}^-=0 \} $ and
					\item $ \forall p \ \in \ S_2=\{ p_j  \in {}^\bullet t_i \ | \ a_{ij}^-=1,b_{ij}^-=1 \} $, $\exists p_k \in S_2$ with $\mu^+(p_k)=1$ and $\forall  p_j \in S_2\backslash \{p_k\}, \quad$ either $\mu^+(p_j)=1$ or $\mu^-(p_j)=1$.
				\end{enumerate} 
				\item \textit{Firing Condition}- An enabled transition $t_i \in T_I$ may fire at a marking $\mu=(\mu^+,\mu^-)$ to yield a new marking $\mu_1$ given by :-
				\\
				\begin{itemize}
					\item $\forall \ p_j \ \in \ S_1,$ 
					$$\mu_1^+(p_j)=\mu^+(p_j)-a_{ij}^- +a_{ij}^+$$
					\item $\forall \ p_j \ \in \ S_2, $ for which
					$\mu^+(p_j)=0 \ \& \ \mu^-(p_j)=1$
					$$\mu_1^+(p_j) =\mu^+(p_j)$$
					$$\mu_1^-(p_j)=\mu^-(p_j)-b_{ij}^- +b_{ij}^+$$
					\item $\forall \ p_j \ \in \ S_2, $ for which $\mu^+(p_j)=1 \ \& \ \mu^-(p_j)=0$\\
					$$\mu_1^-(p_j) =\mu^-(p_j)$$
					$$\mu_1^+(p_j)=\mu^+(p_j)-a_{ij}^- +a_{ij}^+$$
					
					\item  For all $p_j \in P \backslash {}^\bullet t_i$
				\end{itemize}
			\end{itemize}
			$$\mu_1^+(p_j)=\mu_0^+(p_j)-a_{ij}^- +a_{ij}^+$$
			$$\mu_1^-(p_j)=\mu_0^-(p_j)-b_{ij}^- +b_{ij}^+$$
			
			\begin{figure}
				\centering
				\includegraphics[scale=0.35]{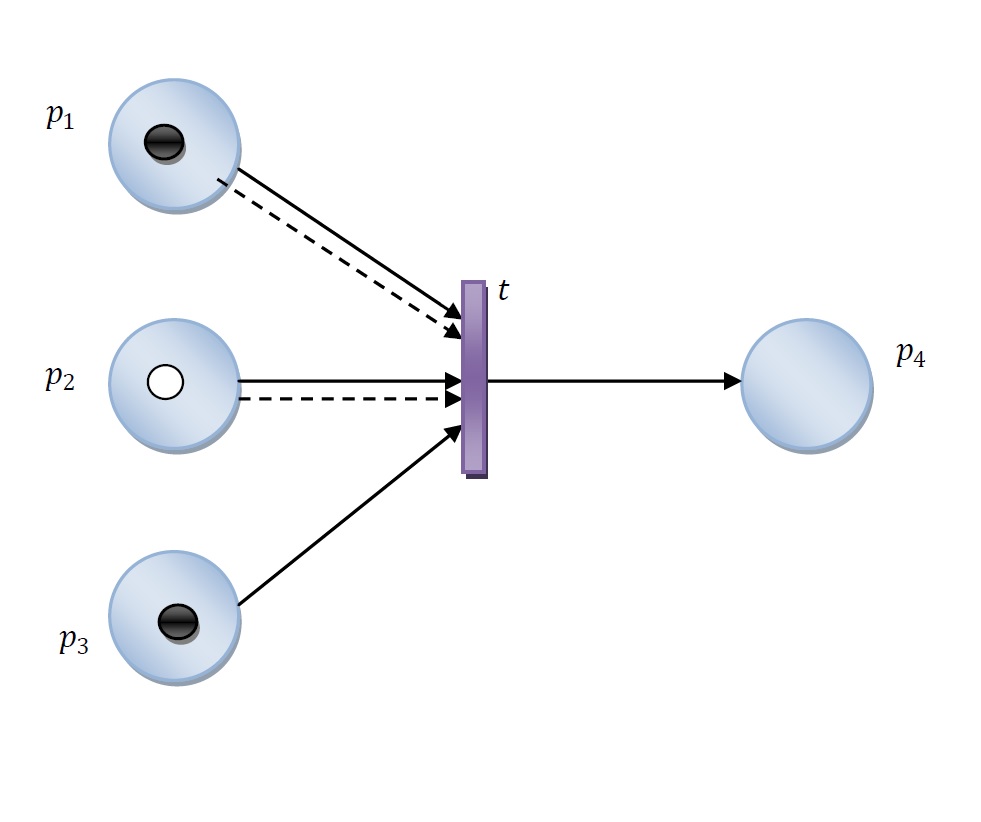}
				\caption{A LSPN with a logic input transition.Transition t is enabled at $\mu=((1,0,1,0),(0,1,0,0))$. t fires to yield a new marking $\mu_1 =((0,0,0,1),(0,0,0,0))$}
			\end{figure}
			\item For $t_i \in T_O$ 
			\begin{itemize}
				\item \textit{Enabling Condition}- A transition $t_i \ \in T_O$ is enabled if  $ \forall p_j \in \ S_1,\mu^+(p_j)=1$.
				
				\item \textit{Firing Condition}-An enabled transition $t_i \in T_O$ may fire at a marking $\mu=(\mu^+,\mu^-)$ to yield a new marking $\mu_1$ given by :-\\
				$\forall p_j \in t_i^\bullet$
				\begin{itemize}
					\item $ \mu_1^+(p_j)=1$ whenever $a_{ij}^+=1,b_{ij}^+=0$ 
					\item If $a_{ij}^+=1,b_{ij}^+=1$ ,then,$\exists p_k \in t_i^\bullet$ with $\mu_1^+(p_k)=1$ and $\forall  p_l \in t_i^\bullet \backslash \{p_k\}$ either $\mu_1^+(p_l)=1$ or $\mu_1^-(p_l)=1$
				\end{itemize}
			\end{itemize}
			$\forall p_j \in P\backslash t_i^\bullet$
			$$\mu_1^+(p_j)=\mu_0^+(p_j)-a_{ij}^- +a_{ij}^+$$
			$$\mu_1^-(p_j)=\mu_0^-(p_j)-b_{ij}^- +b_{ij}^+$$
				\end{enumerate}
\
	\end{enumerate}
	
\end{defi}

\begin{figure}
	\centering
	\includegraphics[scale=0.35]{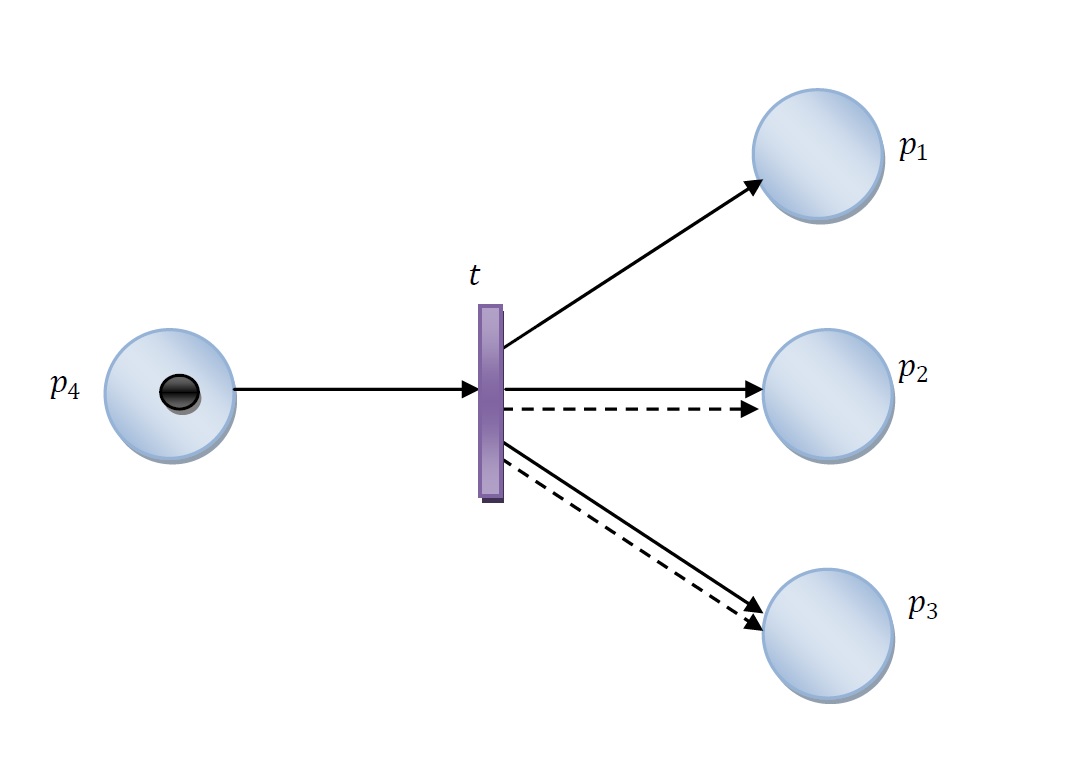}
	\caption{A LSPN with a logic output transition.Transition t is enabled at $\mu=((0,0,0,1),(0,0,0,0))$. t fires to yield a new marking $\mu_1 \ \in \ {((1,1,0,0),(0,0,1,0)),((1,0,1,0),(0,1,0,0)),((1,1,1,0),(0,0,0,0))}$}
\end{figure}

While modeling a workflow process using an LSPN, a positive token in a data place represents that the data has arrived from the organization representing this place, while a negative token represents that the data has not arrived from the organization.So, in a given workflow cycle, all the arrived data is processed while the data which arrives late is processed in the next workflow cycle.

\section{Equivalence of LPN and LSPN}
In order to prove the equivalence between LPN and LSPN, their isomorphism and equivalent definitions are given.
\begin{defi}
	\textbf{Isomorphism}\\
	Let $N'=(P,T,I^-,I^+,I,O,\mu_0)$ be an LPN and $N^*=(P',T',A^-,A^+,B^-,B^+,\mu_0')$ be an LSPN.Let $R(N',\mu_0)$ be the reachability set of $N'$ and $,R^+(N^*,\mu_0')$ be the reachability set of positive markings of  $N^*$.\par 
	Then, $R(N',\mu_0) \ \& \ R^+(N^*,\mu_0')$ are isomorphic if there exists a bijective function $f:R(N',\mu_0) \to R^+(N^*,\mu_0')$such that
	$$\forall \mu_1,\mu_2 \in R(N',\mu_0),t \in T;\mu_1\stackrel{t} \to \mu_2 \implies \exists t' \in T';f(\mu_1)\stackrel{t'} \to f(\mu_2)$$
\end{defi}

\begin{defi}
	\textbf{Equivalence}\\
	Let $N'=(P,T,I^-,I^+,I,O,\mu_0)$ be an LPN and $N^*=(P',T',A^-,A^+,B^-,B^+,\mu_0')$ be an LSPN.Then, $N' \ \& \ N^*$ are equivalent iff $R(N',\mu_0)\ \& \ R^+(N^*,\mu_0')$ are isomorphic to each other.	
\end{defi}

\begin{thm}
	For any Logic Petri Net(LPN) there exists an equivalent Logic Signed Petri Net (LSPN).
\end{thm}
\begin{proof}
	Given an LPN $N'=(P,T,I^-,I^+,I,O,\mu_0)$, we construct an LSPN $ N^*=(P',T',A^-,$\\$A^+,B^-,B^+,\mu_1)$ as follows.
	\begin{enumerate}
		\item The set of places of $N^*$ is same as $N'$ ,i.e. $P'=P$.
		\item Construction of set of transitions $T'$ of $N^*$:
		\begin{enumerate}[(i)]
			\item All the traditional transitions of $N'$ will be the traditional transitions of $N^*$ i.e. if  $t_i \in T_D=T\backslash (T_I\cup T_O) $.Then $t_i \in T'_D \subseteq T'$.This traditional transition $t_i$ in $N^*$ is connected with a place in $N^*$ given by the rule below:
			$$\forall p_j \in P,\ \textnormal{if} \ (p_j,t_i) \in I^- \ \textnormal{then} \ a_{ij}^-=1,b_{ij}^-=0$$
			$$\forall p_k \in P,\ \textnormal{if} \ (p_k,t_i) \in I^+\ \textnormal{then} \  a_{ik}^+=1,b_{ik}^+=0$$
			
			\item All the logic input transitions of $N'$ are logic input transitions of $N^*$ i.e. if $t_i \in T_I =T\backslash (T_D\cup T_O)$.Then $t_i \in T'_I \subseteq T'$ and this transition $t_i \in T_I'$ is connected with a place in $N^*$ by the rule below:\\
			Let ${}^\bullet t_i=\{p_1,p_2,p_3,...,p_k\}$ in $N'$ and $f_I(t_i)$ be the logic input expression associated with $t_i$, convert $f_I(t_i)$ into its disjunctive normal form (DNF), which is unique.If any $p_{j}$ occurs in both the forms i.e. $p_j$ and $\neg p_j$ in the DNF then there exist two arcs in $N^*$ from $p_j$ to $t_i$, one of positive sign and other of negative sign.Therefore, for such $p_j \in {}^\bullet t_i, \ a_{ij}^-=1,b_{ij}^-=1$.
			\\On the other hand, if  $p_{j}$ occurs only in positive form i.e. as $p_j$ and not as $\neg p_j$ in the DNF, then a positive arc is formed from $p_j$ to $t_i $.
			Therefore, for such $p_j \in {}^\bullet t_i, \ a_{ij}^-=1,b_{ij}^-=0$.
			$\forall p_k \in t_i^{\bullet}, \ a_{ik}^+=1 ,b_{ik}^+=0$.

			\item All the logic output transitions of $N'$ are logic output transitions of $N^*$ i.e. if $t_i \in T_O =T\backslash (T_D\cup T_I)$.Then $t_i \in T'_O \subseteq T'$ and this transition $t_i \in T_O'$ is connected with a place in $N^*$ by the rule below:\\
			Let $t_i^\bullet =\{p_1,p_2,p_3,...,p_k\}$ in $N'$ and $f_O(t_i)$ be the logic output expression associated with $t_i$, convert $f_O(t_i)$ into its disjunctive normal form (DNF) which is unique.If any $p_{j}$ occurs in both the forms i.e. $p_j$ and $\neg p_j$ in the DNF then there exist two arcs in $N^*$ from $t_i$ to $p_j$, one of positive sign and other of negative sign.Therefore, for such $p_j \in t_i^\bullet , \ a_{ij}^+=1,b_{ij}^+=1$.
			\\On the other hand, if  $p_{j}$ occurs only in positive form i.e. as $p_j$ and not as $\neg p_j$ in the DNF, then a positive arc is formed from  $t_i $ to $p_j$ .
			Therefore, for such $p_j \in  t_i^\bullet, \ a_{ik}^+=1,b_{ik}^+=0$.\\
			$\forall p_k \in {}^{\bullet}t_i, \ a_{ik}^-=1 ,b_{ik}^-=0$.
		\end{enumerate}
		\item Assignment of tokens 
		\begin{enumerate}
			\item For all $p_j$ satisfying
			$(a_{ij}^-=1$ and $b_{ij}^-=1)$ or $(a_{ij}^+=1$ and $b_{ij}^+=1)$\\
			$\mu_1^+(p_j)=1 \ \& \ \mu_1^-(p_j)=0$ in $N^*$ if $\mu_0(p_j)=1$ in $N'$ and
			
			$\mu_1^+(p_j)=0 \ \& \ \mu_1^-(p_j)=1$ in $N^*$ if $\mu_0(p_j)=0$ in $N'$
			\item For all $p_j$ which has not been assigned a token in step (a) and satisfying either $(a_{ij}^-=1,b_{ij}^-=0)$ or $(a_{ij}^+=1,b_{ij}^+=0)$, \\
			$\mu_1^+(p_j)=1 \ \& \ \mu_1^-(p_j)=0$ in $N^*$ if $\mu_0(p_j)=1$ in $N'$ and\\
			
			$\mu_1^+(p_j)=0 \ \& \ \mu_1^-(p_j)=0$ in $N^*$ if $\mu_0(p_j)=0$ in $N'$
			
		\end{enumerate}
	\end{enumerate}
	%
	The equivalence between $N'$ and $N^*$ can be easily proved, because each marking of $N'$ corresponds to a positive marking in $N^*$ with the same number of positive tokens in the corresponding places in LPN and LSPN (By Construction).That is, in $N'\ \forall \mu_1,\mu_2 \in R(N',\mu_0),t \in T;\mu_1\stackrel{t} \to \mu_2 \implies \exists t' \in T';f(\mu_1)\stackrel{t'} \to f(\mu_2)$
	\par 
	This means that $N'$ and $N^*$ have same behaviour characteristics.Moreover, the structure of $N^*$ is unique since the DNF is unique.So, $f$ is a bijective function and $R(N',\mu_0)\ \& \\ R^+(N^*,\mu_1)$ are isomorphic.
	Consequently, $N'$ and $N^*$ are equivalent.
\end{proof}

\begin{figure}[ht]
	\centering
	\subfloat[LPN with a logic input transition]{{\includegraphics[scale=0.25]{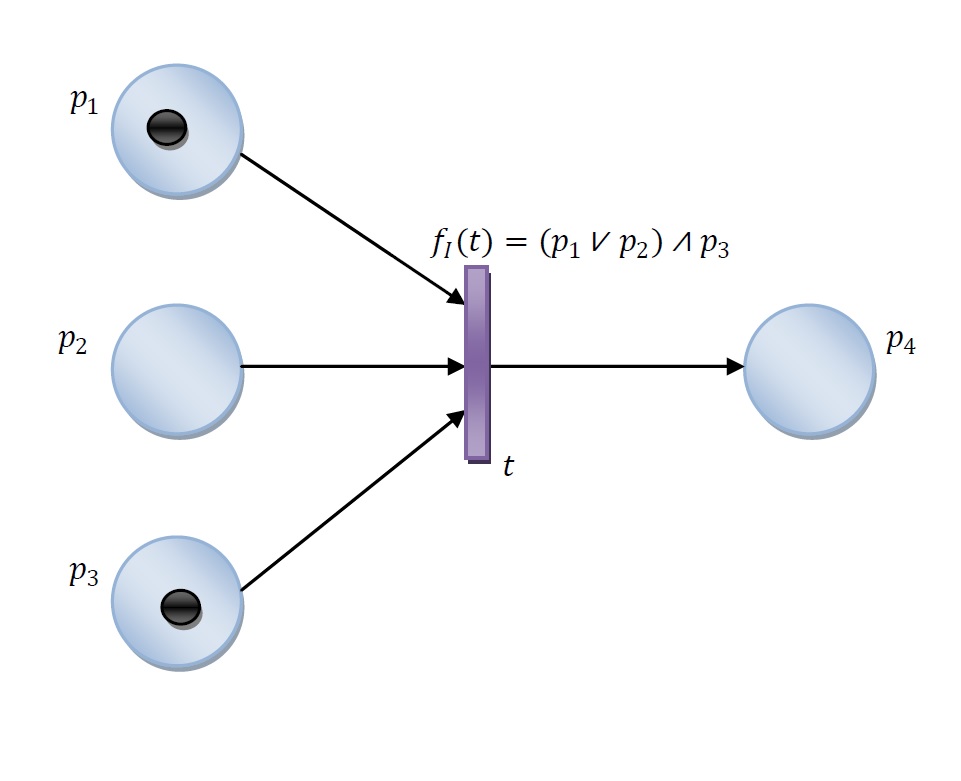} }\label{F1}}
	\qquad \qquad
	\subfloat[Equivalent LSPN for LPN (A)]{{\includegraphics[scale=0.25]{4.jpg} }\label{F12}	}
	\caption{LPN and Equivalent LSPN}
	\label{L1}
\end{figure}
In the figure \ref{L1}, an LPN with a logic input transition is converted to an equivalent LSPN using the procedure in the above theorem.

\begin{figure}[ht]
	\centering
	\subfloat[LPN with a logic output transition]{{\includegraphics[scale=0.25]{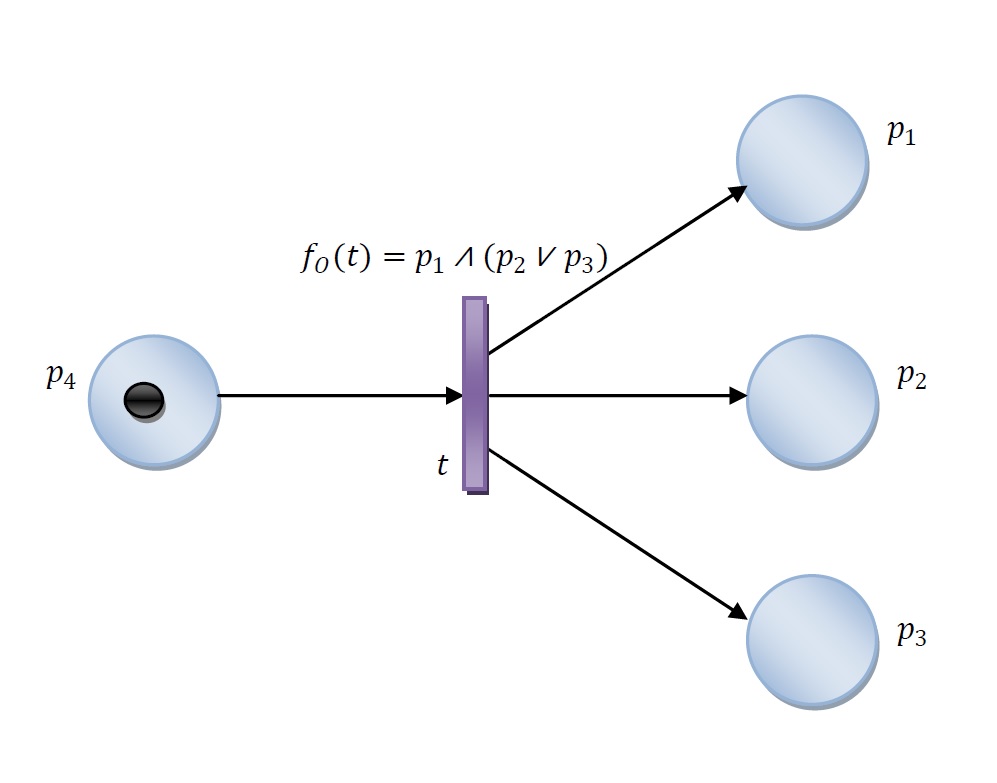} }}
	\qquad   \qquad
	\subfloat[Equivalent LSPN for LPN (A)]{{\includegraphics[scale=0.25]{3.jpg} }}
	\caption{LPN and Equivalent LSPN}
	\label{L2}
\end{figure}
In the figure \ref{L2}, an LPN with a logic output transition is converted to an equivalent LSPN using the procedure in the above theorem.\\

\begin{figure}[!htbp]
	\centering
	\subfloat[AN LPN ]{{\includegraphics[scale=0.3]{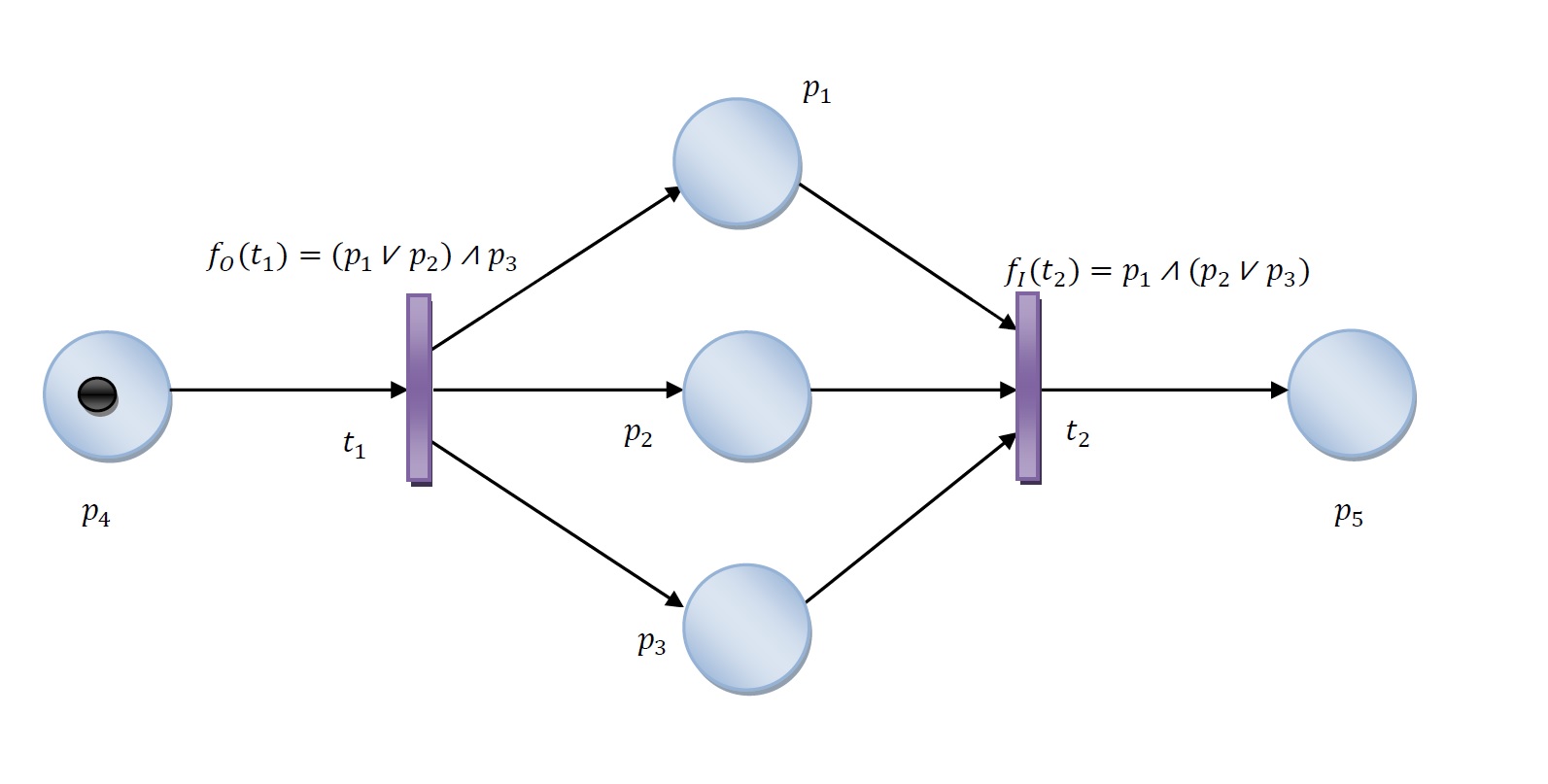} }}
	\qquad \qquad
	\subfloat[Equivalent LSPN ]{{\includegraphics[scale=0.3]{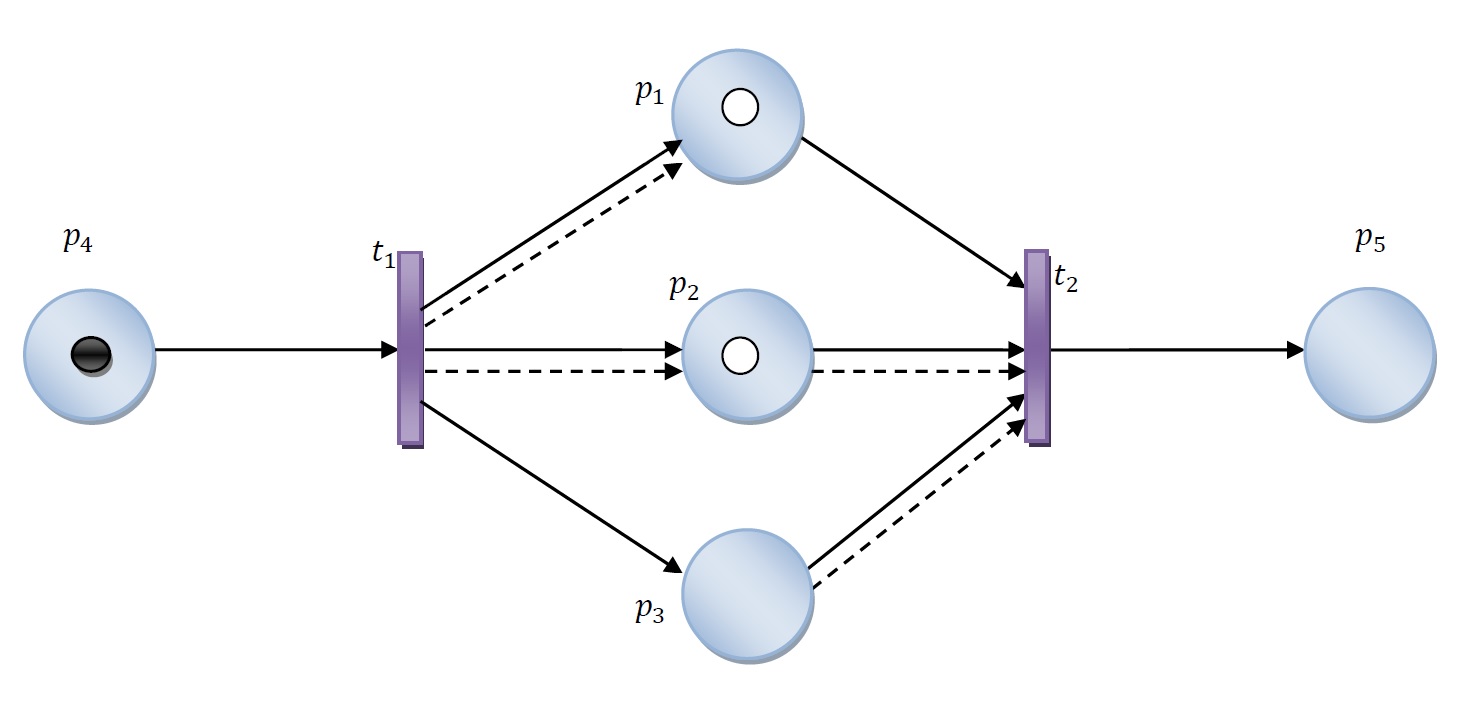} }}

	\caption{Example}
	\label{F2}
\end{figure}
\begin{figure}[h]
	\subfloat[Reachability Tree for LPN]{{\includegraphics[scale=0.35]{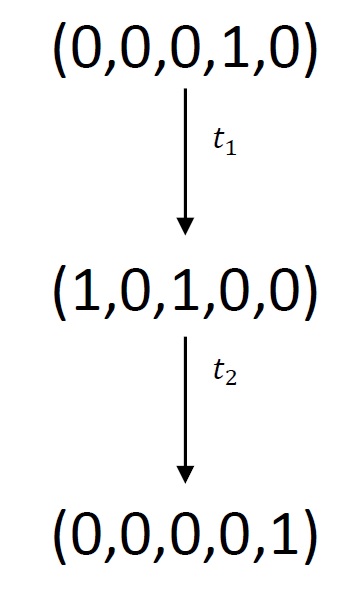} }}
	\qquad   \qquad \qquad \qquad   \qquad    \qquad  
	\subfloat[Reachability Tree for  LSPN]{{\includegraphics[scale=0.35]{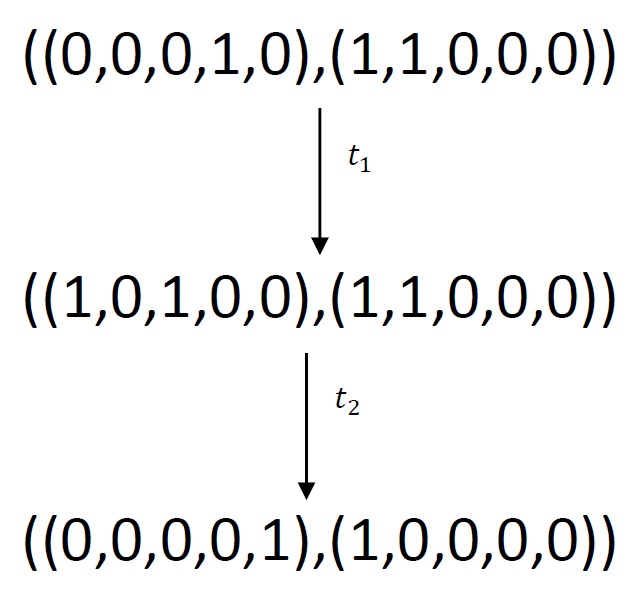} }}
	\qquad
	\caption{Reachability Trees for LSPN Example in Figure \ref{F2}}
	
\end{figure}
In the example in figure \ref{F2} , an LSPN is constructed from a given LPN using the procedure given in theorem above.The reachability tree of LPN and LSPN are also given in figure 7 and it can be clearly seen that corresponding to every marking in LPN we have a positive marking in LSPN.Thus, the reachability set of LPN and the set of positive markings of LSPN are isomorphic which implies that the two PNs are equivalent.
\section{Conclusion}
An LSPN formulated using SPN by mere change of firing rules of the SPN can simulate an LPN.This shows the versatile nature of the SPN introduced by the authors.
%
%
%
%

\bibliographystyle{alpha}
\bibliography{mybibfile}
 \end{document}